\documentclass[11 pt]{article}
\usepackage{graphicx, url}
\usepackage{latexsym}
\usepackage{amsfonts}
\usepackage{amsmath, amsthm}
\usepackage{enumerate}
\usepackage{fullpage}
\usepackage{subfigure}

\newtheorem{theorem}{Theorem}
\newtheorem{lemma}[theorem]{Lemma}
\newtheorem{claim}[theorem]{Claim}
\newtheorem{coro}[theorem]{Corollary}

\begin{document}

\title{The complexity of Shortest Common Supersequence for inputs with no identical consecutive letters}

\author{A. Lagoutte\thanks{LIP, Universit\'e de Lyon, ENS Lyon, CNRS UMR5668, INRIA, UCB Lyon 1, France} \and S. Tavenas\footnotemark[1] 
}

\maketitle

\begin{abstract}
The Shortest Common Supersequence problem ({\tt SCS} for short) consists in finding a shortest common supersequence of a finite set of words on a fixed alphabet $\Sigma$. It is well-known that its decision version denoted [SR8] in~\cite{GJ1979} is NP-complete. Many variants have been studied in the literature. In this paper we settle the complexity of two such variants of SCS where inputs do not contain identical consecutive letters. We prove that those variants denoted {\tt $\varphi$SCS} and {\tt MSCS} both have a decision version which remains NP-complete when $|\Sigma| \geq 3$. Note that it was known for {\tt MSCS} when $|\Sigma| \geq 4$ \cite{FW2012} and we discuss how \cite{Darte2000} states a similar result for $|\Sigma| \geq 3$. 
\end{abstract}

\section{Introduction}

Given two words $u$ and $v$ over an alphabet $\Sigma$, $u$ is a supersequence of $v$ if one can find in $u$ a sequence of non-necessarily successive letters that spells $v$. The shortest supersequence of $u$ is obviously $u$, but the problem becomes more difficult if the input is a set of words and one wants to find a common supersequence for these words as short as possible. The decision version of this problem, called \texttt{SCS} for \texttt{Shortest Common Supersequence} has been proven NP-complete in 1981 by R\"aih\"a and Ukkonen \cite{RU1981}, even if the alphabet has size only 2. It is even NP-complete for some very restricted input, such as in the result of Middendorf \cite{Middendorf1994} on which our work deeply relies: the alphabet is $\Sigma_2=\{0,1\}$, and all the input words have the same length and each contains exactly two non-consecutive 1.
However, another variant of \texttt{SCS}, which we will call \texttt{Modified SCS} (\texttt{MSCS} for short), appears very naturally in the study of combinatorial flood-filling games such as \textsc{Flood-It} and \textsc{Honey-Bee} (studied for example in \cite{Bristol2012, FW2012, FloodIt, MS-FUN2012}). In particular in \cite{FloodIt}, the authors show the NP-completeness of {\tt Flood-It}, using a reduction to {\tt MSCS} with an alphabet of size 3.
This variant is stated as follows:
\medskip

{\bf \noindent Modified Shortest Common Supersequence :\\ {\tt MSCS} - decision version}

{\bf\noindent Input:} A set $L=\{w_1, \ldots, w_n\}$ of words on an alphabet $\Sigma=\{a_1,a_2,\ldots,a_d\}$ such that no word $w_i$ contains two consecutive identical letters and no word $w_i$ starts with letter $a_1$, and an integer $k$.

{\bf \noindent Output:} Does there exist a supersequence of $L$ of size less than $k$?

\bigskip

At first sight, it can seem easy to reduce {\tt SCS} with $\Sigma_2=\{0,1\}$ to {\tt MSCS} with $\Sigma_3=\{0,1,2\}$ by replacing every occurence of $0$ by $02$ in every word of the input set $L$, and doing the reverse operation on the solution to {\tt MSCS} to get the shortest common supersequence of $L$. Unfortunately, this very natural idea does not work in general as shown on the following counter-example. Let $L=\{00111, 11100\}$ be a input for {\tt SCS} (in the minimization version), then the corresponding input for {\tt MSCS} is $L'=\{0202111, 1110202\}$. The shortest solution for {\tt MSCS} is 1110202111 of length 10, and its corresponding candidate solution for {\tt SCS} is 11100111 of length 8. However, the shortest solution for $L$ has size 7: 0011100. The problem here is that the operation transforms 0 into a double-counting letter and looses the symmetry between the two letters. 

The second idea that occurs to mind is then to transform every occurrence of 0 by 02, and also every occurrence of 1 by 12. Then one can hope solving the newly created instance of {\tt MSCS}, and delete every 2 from the solution of {\tt MSCS} to get the shortest solution to {\tt SCS}. This does not work either: consider the instance of {\tt SCS} (in its minimization version) with $L=(\Sigma_2)^3\setminus \{111\}$, that is to say that $L$ contains every word of length 3 on $\Sigma_2=\{0,1\}$ except 111. The shortest supersequence for $L$ is 01010 and is unique. Let $L'$ be the set of words obtained from $L$ by replacing every occurrence of 0 by 02 and every occurrence of 1 by 12. There is no supersequence for $L'$ of length 9 obtained from 01010 by inserting some 2's (there is one of length 10: 0212021202). However there does exist a shortest supersequence for $L'$ of length 9, namely 012012012. Consequently, the very natural ideas do not work for reducing {\tt SCS} to {\tt MSCS}.

Note that Fleischer and Woeginger designed in \cite{FW2012} a reduction proving that {\tt MSCS} is NP-complete when $|\Sigma| \geq 4$ (the conference version of the paper states the result for $|\Sigma|\geq 3$, but the very simple proof turned out to be false; the correct statement appears in the later-published journal version). One should also mention Darte's work \cite{Darte2000} which does not focus directly on {\tt MSCS}, but states a result about typed fusions for typed directed graphs in a compilation context. However, as he explains at the beginning of Section~3.5, when the directed graphs are disjoint union of chains, his problem is equivalent to {\tt SCS}. Moreover the conditions over its typed fusions and digraphs implies that the {\tt SCS} inputs equivalent to his digraph inputs, are words with no identical consecutive letters. Thus Proposition~3 in~\cite{Darte2000} can be interpreted as the fact that {\tt SCS} for inputs with no identical consecutive letters is NP-complete. His reduction from Vertex Cover uses the alphabet $\Sigma=\{0,1,\bar{a}\}$ and a careful look shows that he generates {\tt SCS} inputs where no word starts with $\bar{a}$. Consequently one could state that the NP-completeness of {\tt MSCS} for three letters is shown there. His reduction is derived from a paper of R\"aih\"a and Ukkonen \cite{RU1981} as well as its proof. Unfortunately it is 10 pages long and hard to check.   

The main purpose of our work is to provide a new NP-completeness reduction for {\tt MSCS} when $|\Sigma| \geq 3$, with a shorter proof, so that the result becomes undisputed. To this end, we introduce yet another variant of $\texttt{SCS}$, called {\tt $\varphi$SCS}.
We first define the alphabets $\Sigma_2=\{0,1\}$ and $\Sigma_3=\{0,1,2\}$, and the word morphism $\varphi: \Sigma_2^* \rightarrow \Sigma_3^*$ by $\varphi(0)=0202$ and $\varphi(1)=1$.

\medskip

{\bf \noindent Shortest Common Supersequence for some inputs generated by $\varphi$ : \\ {\tt $\varphi$SCS} - decision version}
\newline {\bf Input:} A set $L=\{w_1, \ldots, w_n\}$ of words on the alphabet $\Sigma_3$ such that $L\subseteq \varphi(\Sigma^*)$, each $w_i$ contains exactly two ones, which moreover are non consecutive, and an integer $k$.

{\bf \noindent Output:} Does there exist a supersequence of $L$ of size less than $k$?

\bigskip

A careful look at those two problems shows that {\tt $\varphi$SCS} is a particular case of {\tt MSCS} if $|\Sigma| \geq 3$. The input words for {\tt $\varphi$SCS} are a concatenation of patterns 0202 and 1 with no consecutive ones, thus they do not contain consecutive identical letters. Moreover none of those input words starts with letter 2. Up to relabelling the letters, one may consider that $a_1=2$. Consequently, we will show that {\tt $\varphi$SCS} is NP-complete, which implies that {\tt MSCS} is also NP-complete if $|\Sigma| \geq 3$. One may wonder why we use the block 0202 instead of the more natural block 02 (that is to say, why $\varphi(0)=0202$ and not $02$). The key reason appears in the third item of Lemma \ref{SCS special} : the elementary technique we use to prove it does not work for the case of blocks 02.

Besides, observe that the threshold on $|\Sigma|$ which involves NP-hardness is tight: when $|\Sigma|=2$, {\tt MSCS} is trivially polynomial.
Finally, let us notice that our proof is a very close adaptation of the proof of Middendorf's result \cite[Theorem 4.2]{Middendorf1994} mentioned in the first paragraph.

\paragraph{Notation} Given two words over an alphabet $\Sigma$, $u=u_1 \ldots u_p$ ($u_i \in \Sigma$) and $v=v_1 \ldots v_q$ ($v_i \in \Sigma$), an {\em embedding} of $u$ into $v$ is an injection $f$ from $\{1,\ldots,p\}$ into $\{1,\ldots,q\}$ such that $u_i=v_{f(i)}$. It tells that $v$ is a \emph{supersequence} of $u$ and we also say that $f$ maps letters of $u$ onto letters of $v$. We will also use equivalently the terms {\em pattern}, {\em block} or {\em factor} to designate a sequence of consecutive letters in a word. A supersequence {\em for a set of words} is a word which is a supersequence for each of those words.

\section{Result}

The NP-completeness reduction will start from Vertex Cover, but we will need the next two lemmas.

\begin{lemma}\label{lemma: 0202}
Let $L$ be a set of words over $\Sigma_3$, such that $L\subseteq \varphi(\Sigma^*)$, and $S=s_1\ldots s_l$ be a supersequence of $L$. Then there exists a supersequence $S'$ of $L$ of size $\leq |S|$ such that $S'\subseteq\varphi(\Sigma^*)$.
\end{lemma}

\begin{proof}
\begin{itemize}
\item \emph{First step:} Let $S'$ be a supersequence of $L$ and let $S''$ be the string obtained from $S'$ after applying one of the following operations:
\begin{enumerate}
\item If $S'$ ends by $0$, delete it.
\item If $S'$ starts by $2$, delete it
\item If $S'$ contains $00$, replace it by $0$.
\item If $S'$ contains $22$, replace it by $2$.
\item If $S'$ contains $01$, replace it by $10$.
\item If $S'$ contains $12$, replace it by $21$.
\end{enumerate}
Then $S''$ is still a supersequence of $L$: indeed, item (i) and (ii) are obvious since no word of $L$ starts by 2 nor ends by 0. For item (iii), observe that no embedding can map two 0 onto two consecutive 0, since no word contains two consecutive 0. Thus if $S'$ contains $00$ at index $i$, and $f$ is an embedding of $w\in L$ so that $f$ maps a zero of $w$ onto $s_{i+1}$, we can modify $f$ to map this zero onto $s_i$. Then $s_{i+1}$ is useless and we can delete it. The same argument applies for item (iv). For item (v), observe that no embedding can map a 0 and a 1 onto two consecutive 0 and 1 because this pattern does not appear in any word of $L$. Thus if $S'$ contains $01$ at index $i$, and $f$ is an embedding of $w\in L$ so that $f$ maps a zero of $w$ onto $s_{i}$ (resp. a one of $w$ onto $s_{i+1}$), we can swap the 0 and the 1 in $S$ and modify $f$ to map the zero of $w$ onto $s_{i+1}$ (resp. the one  of $w$ onto $s_i$). The same argument applies for item (vi).

Consequently, starting from $S$, we can iterately "push" the zeros from left to right by deletion (transformation 00 into 0) or switching (01 into 10), and delete the last letter if it is a zero, until getting a supersequence $S_1$ where each $0$ is followed by a $2$. In the same manner, starting from $S_1$, we can iterately "push" the 2's from right to left until getting a supersequence $S_2$ where each two is preceded by a zero. In other words, $S_2$ is formed by blocks of $02$ and blocks of $1$. Observe that for such supersequences and for every $w\in L$, there always exists an embedding $f$ of $w\in L$ such that for each block $02$, either $f$ maps two consecutive letters to this block or $f$ maps no letter to this block. We will focus only on this type of embedding in the following. Observe moreover that $|S_2|\leq |S|$. 

\item \emph{Second step:} The goal is to build a supersequence $S_3$ formed by blocks of $0202$ and blocks of $1$. 
Suppose first that $S_2$ starts by $(02)^{2k'}1$ for some $k'\in \mathbb{N}$. Consider the first apparition $s_i\ldots s_{i+2k+2}$, $i\in [0:|S_2|]$ of a pattern $1(02)^{2k+1}1$ for any $k\in \mathbb{N}$ and call $2j$ the number of blocks of $02$ before the pattern. Let $S'$ be the string obtained from $S_2$ by replacing this pattern by $1(02)^{2k}102$. Then $S'$ is a supersequence of each $w\in L$: let $f$ be an embedding of $w$ in $S_2$. Either $f$ does not map any letter to $s_{i+2k+2}=1$, or $f$ uses at most $2k$ blocks of $02$ between $s_i=1$ and $s_{i+2k+2}=1$, or there exists a block of $02$ among the $2j$th first blocks such that $f$ maps no letter to this block and $f$ maps no 1 between this block of $02$ and $s_{i+1}$. Otherwise, $w\notin \varphi(\Sigma^*)$. In each one of the three cases, we can easily modify $f$ so that $S'$ is a supersequence of $w$. We can iterate the process until no odd block of $02$ is found. Finally, if $S'$ ends with a pattern $1(02)^{2k+1}$, we can replace it by $1(02)^{2k}$ and still have a supersequence: if $f$ is an embedding of $w\in L$, either $f$ uses only $2k$ blocks among these $2k+1$, or there exists a block of $02$ in $S'$ before the 1 which is not used by $f$ and such that $f$ maps no one after this block. Thus we can modify $f$ as in the previous arguments. The last case if when $S_2$ starts with $(02)^{2k'+1}1$: we can replace this pattern at the very first step by $(02)^{2k'}1{02}$ by the same arguments. Thus we obtain a supersequence $S_3$ of size $\leq |S|$ such that $S_3\in \varphi(\Sigma^*)$.
\end{itemize}
\end{proof}

\begin{lemma}\label{SCS special}
Let $n$ be a positive even integer, $L=\{S_0, \ldots, S_{n^2}\}$ be a set of strings with $S_i=(0202)^{i}1(0202)^{n^2-i}$ for $i\in [0:n^2]$. Then let $S$ be a supersequence of $L$ such that $S\in \varphi(\Sigma^*)$:
\begin{itemize}
\item If $S$ contains exactly $k$ ones, then $S$ contains at least $\lceil(n^2+1)/k\rceil-1+n^2$ blocks of $0202$.
\item If $S$ contains exactly $n^2-1+k$ blocks of $0202$, then $S$ contains at least $\lceil(n^2+1)/k\rceil$ ones.
\item The string $S_{\min}= 1((02)^n1)^{2n}(02)^{n-2}$ is a shortest supersequence of $L$. It has length $4n^2+4n-3$.
\end{itemize}
\end{lemma}

\begin{proof}

\begin{itemize}
\item Let $S$ containing $k$ ones. There must be a subset $L'$ of $L$ which contains at least $\lceil(n^2+1)/k\rceil$ strings such that the strings in $L'$ can be embedded in $S$ in such a way that the ones in these strings are mapped onto the same one of $S$. 
Let $i_{\max}=\max\{i|S_i \in L'\}$ and $i_{\min}=\min\{i|S_i \in L'\}$. Since $S_{i_{\min}}$ and $S_{i_{\max}}$ are mapped onto the same one, $S$ must contain at least $i_{\max}+n^2-i_{\min}$ zeros. 
Moreover, $i_{\max}\geq i_{\min}+\lceil(n^2+1)/k\rceil-1$, so $S$ contains at least $\lceil(n^2+1)/k\rceil-1+n^2$ blocks of $0202$.

\item Let $S$ containing $n^2-1+k$ blocks of $0202$. Consider a one in $S$ and let $L'$ be the subset of $L$ such that the one in the strings of $L'$ is mapped onto this one. 
Let $j$ be the number of blocks of $0202$ before this one in $S$. Then $S_i\in L'$ only if $i\leq j$ and $n^2-i\leq n^2-1+k-j$, i.e. only if $j+1-k\leq i\leq j$. 
Let $i_{\max}=\max\{i|S_i \in L'\}$ and $i_{\min}=\min\{i|S_i \in L'\}$. 
Then $|L'|\leq i_{\max}-i_{\min}+1\leq j-j-1+k+1\leq k$. 
At most $k$ strings are mapped onto the same one, thus there are at least $\lceil (n^2+1)/k\rceil$ ones.
\item $S_{\min}$ is indeed a supersequence of $L$: first, it is a supersequence of $S_0$. Secondly, if $i\neq 0$, there exists $j\in [1:2n]$ such that $(j-1)n/2+1\leq i\leq jn/2$. Then the one in $S_i$ can be mapped to the $(j+1)$th one, and there is $jn/2\geq i$ blocks of $0202$ before the one, and $(2n-j)n/2+n/2-1$ blocks of $0202$ after the one, which is enough to map the suffix $(0202)^{n^2-i}$ because $(2n-j)n/2+n/2-1\geq n^2-((j-1)n/2+1)\geq n^2-i$. So $S_{\min}$ is indeed a supersequence of $L$.

Let $S'$ be the shortest supersequence of $L$. By Lemma \ref{lemma: 0202}, $S'\in \varphi(\Sigma^*)$, so we can apply (i): if $k$ is the number of ones of $S'$, then $|S'|\geq k+4\lceil(n^2+1)/k\rceil -4 +4n^2\geq f(k)$ where $f$ is the function defined on $\mathbb{R}$ by $f(x)=x+4(n^2+1)/x -4 +4n^2$. However, $f$ admits a minimum on $\mathbb{R}$ which is $f(2\sqrt{n^2+1})=4\sqrt{n^2+1}+4n^2-4>4n+4n^2-4$. Consequently,  $|S'|\geq f(k)>4n+4n^2-4$. Since $|S'|$ is an integer, $|S'|\geq 4n+4n^2-3=|S|$.
\end{itemize}
\end{proof}

\begin{theorem}
{\tt $\varphi$SCS} is NP-complete.
\end{theorem}

\begin{proof}
Obviously, {\tt $\varphi$SCS} is in NP. We reduce the Vertex Cover problem to it. Let $G=(V,E)$ be a graph with vertices $V=\{v_1, \ldots, v_n\}$ and edge set $E=\{e_1, \ldots e_m\}$ and an integer $k$ be an instance of Vertex Cover. Recall that the Vertex Cover problem asks whether $G$ has a vertex cover of size $\leq k$, i.e. a subset $V'\subseteq V$ with $|V'|\leq k$ such that for each edge $v_iv_j\in E$, at least one of $v_i$ and $v_j$ is in $V'$. Let us now construct our instance of {\tt $\varphi$SCS}:

For all $i\in [1:n], j\in [0:36n^2]$, let \newline $A_i=(0202)^{6n(i-1)+3n}1(0202)^{6n(n+1-i)}$, \newline $B_j=(0202)^j1(0202)^{36n^2-j}$, \newline
$X_i^j=A_iB_j$.

For each edge $e_l=v_i v_j\in E$, $i<j$, let \newline
$T_l=(0202)^{6n(i-1)}1(0202)^{6n(j-i-1)+3n}1(0202)^{6n(n+2-j)}(0202)^{36n^2-1}$.

Now let $L=\{X_i^j|i\in [1:n], j\in [0:36n^2]\}\cup \{T_l|l\in [1:m]\}$. Clearly, $L$ can be constructed in polynomial time, each string in $L$ is in $\varphi(\Sigma^*)$ and has exactly two ones, which are non consecutive. We will now show that $L$ has a supersequence of length $\leq 168 n^2 +37n-3+k$ if and only if $G$ has a vertex cover $V'$ of size $\leq k$.

Suppose $V'=\{v_{i_1}, \ldots , v_{i_k}\}$ is a vertex cover of $G$. Define \newline
$S'=((02)^{6n}1(02)^{6n})^{i_1-1}1 ((02)^{6n}1(02)^{6n})^{i_2-i_1} 1 \ldots ((02)^{6n}1(02)^{6n})^{i_k-i_{k-1}} 1 ((02)^{6n}1(02)^{6n})^{n+1-i_k} (02)^{6n}$, \newline
$S''=1((02)^{6n}1)^{12n}(02)^{6n-2}$, \newline
$S=S'S''$, then $|S|= 168n^2+37n-3+k$.

By Lemma \ref{SCS special}, $S''$ is a supersequence of $\{B_j|j\in [0:36n^2]\}$. Moreover, $((02)^{6n}1(02)^{6n})^n(02)^{6n}$ is a supersequence of  $\{A_i|i\in [1:n]\}$ thus $S'$ also is. From this we deduce that $S$ is a supersequence of $\{X_i^j|i\in [1:n], j\in [0:36n^2]\}$.

Finally, let us prove that $S$ is a supersequence of $T_l$ for $l\in [1:m]$. Let $e_l=v_iv_j$, $i<j$, and consider the two following cases:
\begin{itemize}
\item \emph{Case 1}: $v_i\in V'$, i.e. there exists $t\in [1:k]$ such that $i=i_{t}$. 
The suffixe $(0202)^{36n^2-1}$ of $T_l$ can be embedded in $S''$. The goal is to prove that the prefix $P_l=(0202)^{6n(i-1)}1(0202)^{6n(j-i-1)+3n}1(0202)^{6n(n+2-j)}$ can be embedded in $S'$.
Observe that one can obtain the following subsequence $S'_l$ of $S'$ by deleting a few ones:  $S'_l=((02)^{6n}1(02)^{6n})^{i_{t}-1} 1 ((02)^{6n}1(02)^{6n})^{n+1-i_{t_1}}(02)^{6n}$. Now $S'_l$ can be rewritten

 $S'_l=((02)^{6n}1(02)^{6n})^{i-1} \underline{1} ((02)^{6n}1(02)^{6n})^{j-i-1}(02)^{6n}\underline{1}(02)^{6n}((02)^{6n}1(02)^{6n})^{n+1-j}(02)^{6n}$. Now we can embed the prefix $P_l$ in $S'_l$ by mapping its two ones onto the two underlined ones of $S'_l$ and checking that the number of blocks of $0202$ is enough.

\item \emph{Case 2}: $v_j\in V'$. The suffixe $(0202)^{3n}(0202)^{36n^2-1}$ of $T_l$ can be embedded in $S''$. We can prove similarly to Case 1 that the prefix $P_l=(0202)^{6n(i-1)}1(0202)^{6n(j-i-1)+3n}1(0202)^{6n(n+1-j)+3n}$ can be embedded in $S'$.

\end{itemize}

Finally, $S$ is a supersequence of $L$ of size $168n^2+37n-3+k$.

Suppose now that $L$ has a supersequence of length $\leq 168n^2+37n-3+k$. By Lemma \ref{lemma: 0202}, $L$ has a supersequence $S$ of size $\leq 168n^2+37n-3+k$ such that $S\in \varphi(\Sigma^*)$. Define $S'$ and $S''$ such that $S=S'S''$, where $S'$ is the shortest prefix of $S$ that contains exactly $6n^2+3n$ blocks of $0202$. Since each $A_i$ contains $6n^2+3n$ blocks of $0202$, like $S'$, and $S$ is a supersequence of $X_i^j$, then $S''$ is a supersequence of $\{B_j|j\in[0:36n^2]\}$. Let us state the following two claims:

\begin{claim}\label{S' number of ones}
For each $i\in[1:n]$, $S'$ must contain a one between the $(6n(i-1)+3n)$th block of $0202$ and the $(6in)$th block of $0202$. Consequently, $S'$ contains at least $n$ ones.
\end{claim}

\begin{proof}
Assume for contradiction that the claim does not hold for an $i\in[1:n]$. Then the one in $A_i$ is mapped on a one in $S$ which is after at least $6in$ blocks of $0202$. Since $S'$ contains only $6n^2+3n$ blocks of $0202$, the suffix $(0202)^{3n}$ of $A_i=(0202)^{6n(i-1)+3n}1(0202)^{6n^2+3n-6ni}(0202)^{3n}$ must be mapped onto $S''$. Consequently, $S''$ is a supersequence of $\{(0202)^{3n}B_j|j\in[0:36n^2]\}$, thus by Lemma \ref{SCS special}, $|S''|\geq 4\cdot 3n+144n^2+24n-3=144n^2+36n-3$. Since $|S'|\geq 4(6n^2+3n)$, we have $|S|\geq 24n^2+12n+144n^2+36n-3=168n^2+48n-3>168n^2+37n-3+k$, a contradiction.
\end{proof}

\begin{claim}\label{S' vertex cover}
For $l\in[1:m]$ and $e_l=v_iv_j$, $i<j$, $T_l$ cannot be embedded in $S$ if $S'$ contains a one neither between the $6n(i-1)$th block of $0202$ and the $(6n(i-1)+3n)$th block of $0202$, nor between the $6n(j-1)$th block of $0202$ and the $(6n(j-1)+3n)$th block of $0202$.
\end{claim}

\begin{proof}
Assume that the claim does not hold for an $l\in[1:m]$ with $e_l=v_iv_j$, $i<j$. The suffix $(0202)^{6n+36n^2-1}$ of $T_l$ must be mapped onto $S''$: indeed, let $P_l=(0202)^{6n(i-1)}1(0202)^{6n(j-i-1)+3n}1(0202)^{6n(n+1-j)}0$ be a prefix of $T_l$. The first one (resp. second one, last zero) of $P_l$ must be mapped to a one (resp. one, zero) of $S$, let $t_1$ (resp. $t_2$, $t_3$) be the number of blocks of $0202$ in $S$ before this one (resp. one, zero). The assumption implies $t_1\notin [6n(i-1):6n(i-1)+3n]$ and $t_2\notin [6n(j-1):6n(j-1)+3n]$. By definition of $P_l$, $t_1\geq 6n(i-1)$ thus, by assumption $t_1\geq 6n(i-1)+3n$. By definition of $P_l$ again, $t_2\geq t_1+6n(j-i-1)+3n\geq 6n(j-1)$. Consequently, $t_2\geq 6n(j-1)+3n$. Finally, $t_3\geq t_2+6n(n+1-j)\geq 6n^2+3n$. Since $S'$ contains exactly $6n^2+3n$ blocks of $0202$, the last zero of $P_l$ is mapped onto $S''$.

Consequently, $S''$ must contain at least $6n+36n^2-1$ blocks of $0202$. Assume $S''$ contains $36n^2+p$ blocks of $0202$ with $p\geq 6n-1$. By Lemma \ref{SCS special} (ii), $|S''|\geq 4(36n^2+p)+\lceil(36n^2+1)/(p+1)\rceil\geq f(p)$ where $f$ is the function defined on $\mathbb{R}$ by $f(x)=4(36n^2+x)+(36n^2+1)/(x+1)$. But $f$ is increasing on $[3n:+\infty [$. Since $p\geq 6n-1$, $f(p)\geq f(6n-1)=4(36n^2+6n-1)+(36n^2+1)/(6n)>144n^2+24n-4+6n$. Thus $|S''|\geq 144n^2+30n-3$. 

Since $S'$ contains $6n^2+3n$ blocks of $0202$ and, as a consequence of Claim \ref{S' number of ones}, at least $n$ ones, we have $|S'|\geq 4(6n^2+3n)+n=24n^2+13n$. Consequently, $|S|\geq 144n^2+30n-3+24n^2+13n=168n^2+43n-3>168n^2+37n-3+k$, a contradiction.
\end{proof}

\emph{Conclusion} By Lemma \ref{SCS special}, $|S''|\geq 144n^2+24n-3$. By definition, $S'$ contains $6n^2+3n$ blocks of $0202$. Since $|S|\leq 168n^2+37n-3+k$, $S'$ can contain at most $168n^2+37n-3+k-(144n^2+24n-3)-4(6n^2+3n)=n+k$ ones. By Claim \ref{S' number of ones}, there is a one between the $(6n(i-1)+3n)$th zero and the $6in$th zero of $S'$ for each $i\in[1:n]$, which makes $n$ ones. This implies that there can be at most $k$ indices $i\in[1:n]$ such that there is a one between the $6n(i-1)$th zero and the $(6n(i-1)+3n)$th zero of $S'$. Let $i_1, \ldots i_p$ be these indices, $p\leq k$. Thanks to Claim \ref{S' vertex cover}, we see that $\{v_{i_1}, \ldots, v_{i_k}\}$ is a vertex cover of $G$ of size $p\leq k$.

\end{proof}

As explained in the presentation of the two variants, inputs for {\tt $\varphi$SCS} have no identical consecutive letters and do not start with~2, thus we immediatly gain the following corollary.

\begin{coro}
{\tt MSCS} is NP-complete when $|\Sigma| \geq 3$.
\end{coro}

\bibliographystyle{plain}
\bibliography{FLOODS}

\end{document}